\pdfoutput=1
\documentclass[runningheads]{llncs}
\usepackage[utf8]{inputenc}
\usepackage[x11names,dvipsnames,table]{xcolor} 
\usepackage{mathtools}
\usepackage{amsmath,amsfonts,amssymb}
\usepackage{enumitem}
\usepackage{graphicx,subfigure}
\usepackage{wasysym}
\usepackage{algpseudocode,algorithm}
\usepackage{capt-of}
\usepackage{cleveref}
\usepackage{amsfonts}
\usepackage{array,multirow}
\usepackage{placeins}
\usepackage{wrapfig}
\usepackage{placeins}
\usepackage{mathrsfs}
\usepackage{sidecap, caption}

\DeclareMathAlphabet{\mathpzc}{OT1}{pzc}{m}{it}

\usepackage{lineno}
\usepackage{flushend}

\begin{document}
	\newcommand{\etal} {{\it et al.}\xspace}

	\title{Large $k$-gons in a 1.5D Terrain}
	
	\author{Vahideh Keikha}
	
	\institute{Institute of Computer Science, the Czech Academy of Sciences, Department of Applied Mathematics, Charles
		University, Prague, Czech Republic, \email{keikha@cs.cas.cz}
	}
	\authorrunning{V.Keikha}
	\maketitle 
	%\linenumbers

	\begin{abstract}
		Given is a 1.5D terrain $\mathcal{T}$, i.e., an $x$-monotone polygonal chain in $\mathbb{R}^2$. For a given $2\le k\le n$, our objective is to approximate the largest area or perimeter convex polygon of exactly or at most $k$ vertices inside $\mathcal{T}$. For a constant $k>3$, we design an FPTAS that efficiently approximates the largest convex polygons with at most $k$ vertices, within a factor $(1-\epsilon)$. 
		For the case where $k=2$, we design an $O(n)$ time exact algorithm for computing the longest line segment in $\mathcal{T}$, and for $k=3$, we design an $O(n \log n)$ time exact algorithm for computing the largest-perimeter triangle that lies within $\mathcal{T}$.

	\end{abstract}

	\section{Introduction}
	Recently, much attention is given to problems of the following form: Given is
	a set  $\mathcal{P}$ of complexity $n$ and a parameter $k \le n$. The objective is computing a  subset of size $k$ of  $\mathcal{P}$ 
	that optimizes some cost function among all choices. 
	
	A well-studied variant is ``potato peeling''~\cite{goodman1981largest} also known as { ``convex skull''} \cite{woo1986convex}, which assumes the input set $\mathcal{P}$ is an arbitrary simple polygon of $n$ vertices, and seeks the inscribed object having the largest area or perimeter in $\mathcal{P}$.  %In this paper, we address this problem by assuming the input object is a 1.5D terrain, i.e., an $x$-monotone polygonal chain in $\mathbb{R}^2$. 
	It is shown that computing an inscribed convex polygon having the largest area in $\mathcal{P}$ is  solvable in $O(n^7)$ time~\cite{ChangY86}. 
	In the same paper, an $O(n^6)$ time algorithm is also presented for computing the largest-perimeter inscribed convex polygon. 
	There also exists a 4-approximation algorithm for the largest-area contained convex polygon in $\mathcal{P}$ with running time $O(n\log n)$~\cite{hall2006finding}. 
	In the same paper, the authors also discussed several approximations and randomized algorithms for the other variants of the problem, including the largest-area rectangle, ellipse, and the $\delta$-fat triangle in $\mathcal{P}$, where an  $\delta$-fat triangle is a triangle at which all three angles are at least $\delta$. %A $(1-\epsilon$)-approximation of its area can be computed in $O(n)$ time. 
	In~\cite{CabelloCKSV17},  computing the largest inscribed convex polygon is studied for both the area and the perimeter measures, and the authors introduced a   randomized  $(1- \epsilon)$-approximation algorithm that gives this result  with probability at least $2/3$,  and runs  in $O(n(\log^2 n+(\frac {1} {\epsilon^3})\log n+ \frac {1}{\epsilon^4}))$ time. 
	
	The exact algorithm for computing the largest area triangle that is inscribed in $\mathcal {P}$ is also studied in \cite{MelissaratosS90}. The authors presented an $O(n^4)$ time algorithm that also works for the perimeter measure keeping the same running time. %Hence, these problems seem to be easier in a terrain. 
	
	In~\cite{DanielsMR97}, an $O(n \log^2 n)$ time algorithm is presented for computing the largest area axis-aligned rectangle in a simple polygon, and it is shown that this problem has an $\Omega (n \log n)$ lower bound. See~\cite{karmakar2016construction} and the references therein for the results on the potato peeling problem in higher dimensions. 
	
	\begin{figure}[t]
		\centering
		\includegraphics[scale=0.75]{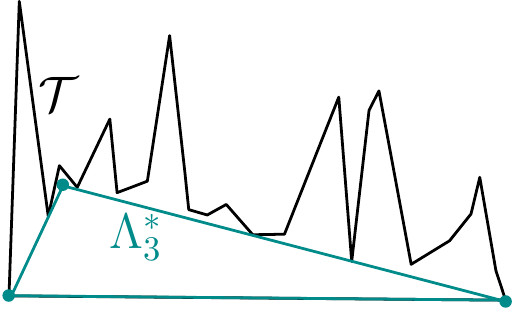}
		\caption{ A terrain and its largest-perimeter triangle  $\Lambda^*_3$. In this example,   $\Lambda^*_4$ is the same as  $\Lambda^*_3$. Because any convex polygon in $\mathcal{T}$ of strictly four vertices has a perimeter smaller than the perimeter of $\Lambda^*_3$. }
		\label{fig:def1}
	\end{figure}
	
	\subsubsection{Motivation} 
	There are several applications for different variants of the potato peeling problems. One of the most important applications is {\em collision detection}~\footnote{also known as {\em interference detection} or {\em contact determination}.} with the aim of finding the biggest object of bounded complexity which could be moved within an environment without being collapsed~\cite{lin1998collision}. Computing the biggest static or moving object in an environment is considered a major computational bottleneck. 
	The other application is {\em geometric shape approximation} which is the act of approximating one geometric shape with another simpler object, see, e.g.,~\cite{alt1990approximation} which gives several approximations of a convex polygon with rectangles, circles, and polygons with fewer edges. See also~\cite{luebke2001developer} for a survey on different applications and approaches to these problems.

	\subsubsection*{Related Work}
	The longest line segment inside a simple polygon of $n$ vertices can be computed in $O(n^{1.99})$ time~\cite{chazelle1990algorithm} and can be approximated within a factor $(1-\epsilon)$ in $O(n \log^2 n)$ time~\cite{hall2006finding}. 
	As mentioned above, the best-known algorithms for computing the largest area and the largest perimeter triangle inside a simple polygon take $O(n^4)$ time~\cite{MelissaratosS90}. However, these problems seem to be easier if the input is a terrain.  
	It is shown that the largest-area triangle that is inscribed in a 1.5D terrain $\mathcal{T}$ of $n$ vertices can be computed in $O(n^2)$ time~\cite{das2021largest}. This algorithm recently improved to $O(n \log n)$ time~\cite{cab}.  
	In~\cite{das2021largest}, a $2$-approximation algorithm and an FPTAS with running times $O(n \log n)$ and $O(\frac n \epsilon  \log ^2 n)$, respectively, are designed for computing the largest-area triangle. 
	The largest area axis-aligned rectangle inside  $\mathcal{T}$ can also be computed in $O(n)$ time~\cite{daniels1995finding} since any terrain can be considered as a vertically separated horizontally convex polygon~\cite{daniels1995finding}. 
	%To the best of our knowledge, the other objects contained in a 1.5D terrain are not considered thus far. 

	\subsubsection*{Contribution}
	In the following, we formally define our problems. 
	\begin{problem} \label{prob:triangle}
		Given is a 1.5D terrain $\mathcal{T}$, compute the largest perimeter triangle in~$\mathcal{T}$.  
	\end{problem} %\vspace{-0.3cm}
	\begin{problem} \label{prob:kgon}
		Given is a 1.5D terrain $\mathcal{T}$ and a constant integer $2 \le k<n$, compute a largest perimeter/area polygon of at most $k$ vertices in $\mathcal{T}$.   
	\end{problem}
	See Fig.~\ref{fig:def1} for an illustration. For the case where $k=2$, we seek the longest line segment inside the terrain and refer to it as the diameter of the terrain. We show that the diameter of the terrain can be computed in $O(n)$ time.  
	For $k=3$, we design an $O(n \log n)$ time exact algorithm for the largest perimeter contained triangle in $\mathcal{T}$, that matches the result for the area measure~\cite{cabello2016finding}. We then give a $(1-\epsilon)$-approximation for Problem~\ref{prob:kgon}, i.e.,  computing the largest perimeter polygon of at most $k>3$ vertices.  
	Our results are successfully extended to the area measure.  See  Table~\ref{table:results} for a summary of the new and known results. %In this note we focus on computing the largest-perimeter $k$-gons inside a 1.5D terrain. 
	%Table~\ref{table:results} gives a summary of the new and known results. 

	\subsubsection*{Preliminaries}
	Let $\mathcal{T}$ be a 1.5D terrain of $n$ vertices, and let $\mathcal{B}$ be the base of the terrain. Let $\Lambda^*_k$ denote a largest perimeter polygon of at most $k$ vertices inside $\mathcal{T}$. 
	W.l.o.g, suppose $\mathcal{B}$ is always a horizontal segment, otherwise one can transform the coordinate system to satisfy this. Let $Q$ be any convex polygon inside $\mathcal{T}$. 
	The sides of $Q$  which have exactly one endpoint on $\mathcal{B}$ are called the {\em legs} of $Q$, and the side which has two vertices at $\mathcal{B}$  is called the {\em base} of $Q$.  For a point $p \in \mathcal{T}$,  $x(p)$ denotes the $x$-coordinate of $p$.

	%For two convex polygons $Q_1,Q_2$, $Q_1 \ge Q_2$ means  the area (or the perimeter) of $Q_1$ is than $Q_2$. 

	{\footnotesize{
			\begin{table}[t]
				\centering
				\begin{tabular}{|p{3.3cm}|c|c|c|c|c|}
					\hline
					Problem & Time & Apprx. & In Object & Ref.\\
					\hline\hline
					Longest line segment  &  $O(n ^{1.99})$ & exact & Simple Polygon & \cite{chazelle1990algorithm} \\
					Longest line segment  &  $O(n \log^2 n)$ & $1-\epsilon$ & Simple Polygon & \cite{hall2006finding} \\
					Max $\mathscr{A}$ Triangle  &  $O(n^4)$& exact & Simple Polygon &\cite{MelissaratosS90} \\ 
					Max $\mathscr{P}$  Triangle  &  $O(n^4)$& exact & Simple Polygon &\cite{MelissaratosS90} \\
					Max $\mathscr{A}$ Triangle  &  $O(n \log n)$& 1/4 & Simple Polygon &\cite{hall2006finding} \\
					Max $\mathscr{A}$ Triangle  &  $O(n \log n)$& exact & Terrain &\cite{cab} \\
					Max $\mathscr{A}$ Triangle  &  $O(\epsilon^{-1} n \log^2 n)$& $1-\epsilon$ &  Terrain& \cite{das2021largest} \\ \hline
					Longest line segment  &  $O(n)$& exact & Terrain & Thm.~\ref{thm:diameteroft} \\
					Max $\mathscr{P}$ Triangle  &  $O(n  \log n)$& exact & Terrain & Thm.~\ref{thm:exacttri} \\	
					% Max $\mathscr{P}$ Triangle  &  $O(\epsilon^{-6} n \log^2 n)$& $1-\epsilon$ & Terrain & Thm.~\ref{thm:tri} \\ 		
					Max $\mathscr{P}$ at most $k$-gon  &  $O(k^{10}\epsilon^{-2k} n \log^2 n)$& $1-\epsilon$ & Terrain & Thm.~\ref{thm:kgon} \\	
					Max  $\mathscr{A}$ at most $k$-gon  &  $O(k^{10}\epsilon^{-2k} n \log^2 n)$& $1-\epsilon$ & Terrain & Thm.~\ref{thm:areatri}	\\				
					\hline
				\end{tabular}
				\caption{A summary of the new and known results; $\mathscr{A}$, $\mathscr{P}$ stands for the area and perimeter, respectively. }
				\label{table:results}
			\end{table}
	}}

	%\section{Diameter of the Terrain} 
	
	\section{Case $k=3$: Largest Perimeter Triangle}  \label{sec:perimeter}
	
	In contrast to the largest area triangle inside a terrain, a largest perimeter triangle in a terrain does not necessarily have a side coincident with $\mathcal{B}$. See Fig.~\ref{fig:translate}a. 
	We show that there exists a $\Lambda^*_3$ which has a vertex or a side coincident with $\mathcal{B}$. 
	We consider these two cases independently, and report the best solution. 
	%The same observation is also proved for the largest area triangle in $\mathcal{T}$~\cite{das2021largest}.  

	\begin{figure} 
		\centering
		\includegraphics[scale=0.7]{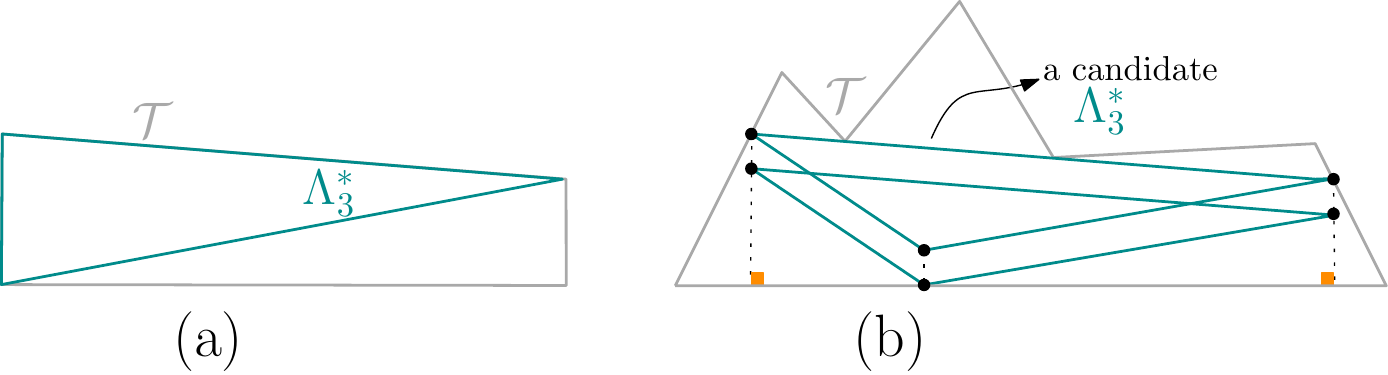}
		\caption{(a) A terrain of four vertices at which $\Lambda^*$ does not have a side on the base of the terrain. (b) Illustration of Lemma~\ref{lem:base}.  }
		\label{fig:translate}
	\end{figure}
	
	\sidecaptionvpos{figure}{c}
	\begin{SCfigure} [40][t]
		%\centering
		\includegraphics[scale=0.9]{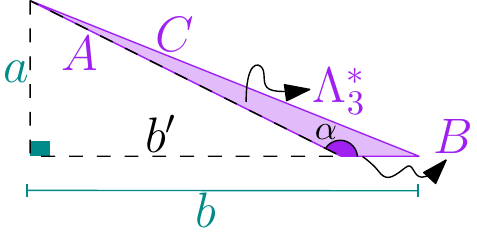}
		\caption{Illustration of Lemma~\ref{lem:angpstar}.  }
		\label{fig:pstar}
	\end{SCfigure}

	\begin{lemma} \label{lem:base}
		There is a largest perimeter triangle inside the terrain which has its base or one of its vertices coincident with the base of the terrain. 
	\end{lemma}	
	\begin{proof}
		 If $\Lambda^*_3$ does not have any vertex on the base of the terrain, since the terrain is monotone, we can translate $\Lambda^*_3$ downward until it has a vertex or a side on the base of the terrain. See Fig.~\ref{fig:translate}b.  \qed
		\end{proof}

 %The same argument also holds for the area.  
 \subsubsection*{$\Lambda^*_3$ having its base on $\mathcal{B}$ } \label{sec:exactalg}  
	We first seek a $\Lambda^*_3$ which has its base coincident with $\mathcal{B}$. We will use the following lemmas in some of our proofs. 
	
	\begin{lemma}
		\label{lem:angpstar}
		Let the base of  $\Lambda^*_3$ coincides with $\mathcal{B}$, and let $\alpha$ and $\beta$ be the internal angles of the legs of $\Lambda^*_3$. Then  $0<\alpha,\beta \le \frac \pi 2$. 
	\end{lemma}
	
	\begin{proof}Suppose, by contradiction, there is a solution to $\Lambda^*_3$ at which (w.l.o.g) $\alpha > \frac \pi 2$.  Let $A,B$ and $C$  denote the sides of $\Lambda^*_3$ in counter clockwise (CCW), where $A$ and $C$ are the legs of $\Lambda^*_3$, as illustrated in Fig.~\ref{fig:pstar}.
		With a slight abuse of the notation, let each side also denote its length. 
		Consider the right triangle $\Lambda'$ which is constructed by drawing the perpendicular to the base of the terrain through the top vertex of $\Lambda^*_3$. Let $a,b$ and $C$ denote the sides of $\Lambda'$, where $C$ denotes the hypotenuse of this triangle and $a,b$ denote the adjacent sides of the right angle. Let  $b'= b- B$.   
		From the triangle inequality, we have $a+b'>A$. By adding a $B$ to either side of the inequality we have $a+b'+B>A+B$, hence $a+b>A+B$, which means the perimeter of $\Lambda'$ is larger than the perimeter of $\Lambda^*_3$, a contradiction. 
		Consequently, we have
		$0<\alpha,\beta \le \frac \pi 2$. \qed
	\end{proof}

	\begin{lemma}\label{lem:alg}
		A largest perimeter triangle which has its base on $\mathcal{B}$, either (i) has its topmost vertex on the boundary of the terrain, in which case each of its legs is passing through at least a vertex of the terrain, or (ii) each of its sides is supported by two vertices of the terrain. 
	\end{lemma}

	\begin{proof}
		The sides or the vertices of  $\Lambda^*_3$ must be blocked by the boundary of the terrain, otherwise, we still can improve the perimeter. 
		
		In the following, we prove for restricting any further possible enlargement, $\Lambda^*_3$ must have either its top  vertex on the boundary of the terrain in which case each of its legs is passing through a vertex,
		or it has its legs passing through two vertices of the terrain. 
		In case (i), for the sake of contradiction, suppose $\Lambda^*_3$ has its top vertex on the terrain, but the legs of $\Lambda^*_3$   are not supported by any vertices of the terrain. But then, while we keep the top vertex fixed, we improve the perimeter by moving the vertices on the base of the terrain in opposite directions to get further from each other,  until each leg encounters a vertex of the terrain which restricts its further movement. %See Figure~\ref{}. 
		This gives a contradiction with the optimality of $\Lambda^*_3$. 
		
		In case (ii),  suppose there is a solution to $\Lambda^*_3$ where a leg $\ell$ of $\Lambda^*_3$ is supported by only one vertex (otherwise, while the top vertex is fixed, similar to above, we can improve the perimeter by moving the vertices of $\Lambda^*_3$ on the base on opposite directions until each leg touches at least one vertex of the terrain). 
		From Lemma~\ref{lem:angpstar}, the interior angles of $\Lambda^*_3=ABC$ are at most $\frac {\pi}{2}$. Now suppose, by contradiction, at least one of the legs of $\Lambda^*_3$, say the right one, is supported by only one vertex.
		Let $l_r$ denote the right leg, and let $l_l$ denote the left leg. 
		Observe that if $l_r$ is supported by a convex vertex, the top vertex of $\Lambda^*_3$ lies at the boundary of the terrain. Hence let 
		$l_r$ be supported by a reflex vertex $O$. We consider the changes of the perimeter of $\Lambda^*_3$ as a function of $l_r$, while we rotate $l_r$ around $O$ in CW or CCW directions. %with keeping $l_l$ unchanged.
		
		%We use basic geometry to prove that in at least one of the CW or CCW directions we rotate $l_r$ around $O$, the perimeter would be increased. %Consider the segment $\ell$ that has its endpoints at $l_l$ and $\mathcal{B}$ and passes through $O$. Let $\ell$ be pivoted at $O$. 
		$l_r$ can be rotated around $O$ until it encounters another vertex at which $l_r$ is blocked and further rotation is not possible. 
		
		\iffalse Let $[b_0,b_{\ell}] \in l_l$ denotes the range for which $\ell$ is pivoted at $O$ and is rotating in CW and CCW.  
		
		W.l.o.g, let $O=(0,0)$. We assume the endpoint of $\ell$ at $l_l$ is a variable $t \in [0,1]$, where it starts moving from $b_0$ (corresponding to the value 0 for $t$) toward $b_{\ell}$ (corresponding to the value 1 for $t$). 
		
		Hence, the length of $\ell$ is a hyperbolic function, where its minimum in the range $[b_0,b_{\ell}]$ is unique and can be computed precisely.
		On the other hand, the maximum of this function happens at the border values of $t$, for which $\ell$ passes through two vertices of the terrain, where one is $O$ and the other is either $b_0$ or $b_{\ell}$. The Lemma follows. 
		%and occurs when $O$ lies at the midpoint of $l_r$.  
		But we are interested in maximizing this function. \fi

		First, suppose the point $O$ lies on the same side  as $l_l$, with reference to the supporting line through the bisector of $l_l$ and $\mathcal{B}$ (the case where the point $O$ lies on the other side of the bisector is symmetric); see Fig.~\ref{fig:hyperbol}. Consider two triangles $Ab_2c_1$ and $Ab_0c_3$ by a slight rotation of $l_r$ in CW and CCW direction, respectively, where $b_0B=Bb_2$. 
		Note that such triangles always exist even for a very small value of $b_0B$, as $l_r$ is pivoted at only one vertex. Suppose by rotating $l_r$, around $O$, in the CW direction, the perimeter of the resulting triangle is decreased. Hence, 
		$Ab_2c_1 < ABC$~(1). We show that further rotations of $l_r$ in the CCW direction would increase the perimeter of $ABC$.

		Consider $b_3c_0$ as a further rotation of $l_r$ in CW direction, at which $c_0$ lies at $\mathcal{B}$,  and $b_3$ lies at the supporting line of $AB$. Let $Oc_o$ is vertical to $\mathcal {B}$.
		
		From (1) we have $Ab_0+b_0B+Bb_2+b_2c_1+Ac_0+
		c_0c_1 < Ab_0+b_0B+BC+c_1C+c_0c_1+Ac_0$. Since $b_0B=Bb_2$, we have $Bb_2+ b_2c_1 <BC+c_1C$.

	\begin{SCfigure} [40][t]
			\centering
			\includegraphics[scale=0.6]{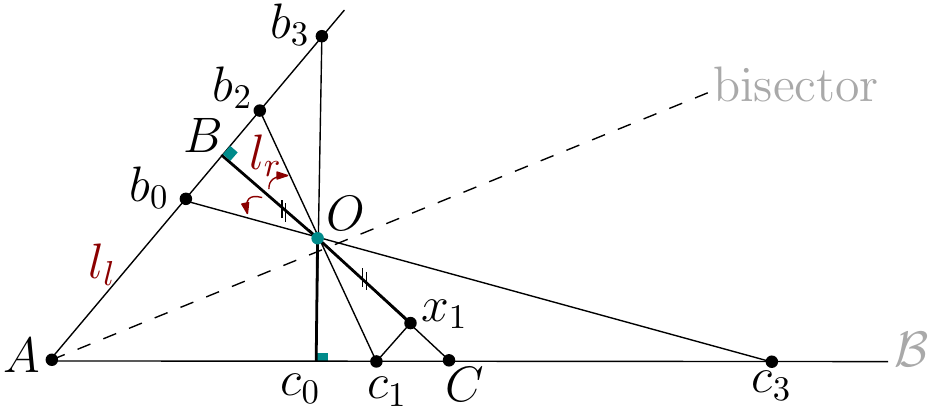}
			\caption{Illustration of Lemma~\ref{lem:alg}. }
			\label{fig:hyperbol}
		\end{SCfigure}

		First let $BC$ is vertical to $l_l$. 
		We show that if we rotate $l_r$ in CCW, the resulting triangle is bigger than $BC$ and has an internal angle bigger than $\frac {\pi}{2}$, which cannot happen according to Lemma~\ref{lem:angpstar}. If $BC$ is vertical to $l_l$, the rotation of $l_r$ in the CW direction is allowed, but it should decrease the perimeter of the resulting triangle. But then, we show that the assumption of by rotating $l_r$ in CW direction the perimeter of the resulting triangle is decreased gives a contradiction with the optimality of $ABC$. 
		
		If $O$ lies on the bisector of the angle $A$, the triangles $ABO$ and $AOc_0$ have the same perimeter, while $Oc_o$ is vertical to $\mathcal {B}$. The same argument also holds for $OBb_0$ and $Oc_0c_1$. But if $O$ lies to the left of the bisector of angle $A$, directed from $A$ to $O$, the triangles are still similar, but $ABO < AOc_0$ and  $b_0BO  < Oc_0c_1$, since the length of one side of each of these triangles, is decreased. 
		Consider a line parallel to $AB$ through $c_1$. Let $x_1$ be the intersection of this parallel line and $BC$.

		Since $O$ lies to the left of the bisector of angle $A$, directed from $A$ to $O$, 	from and the similarity of the triangles $OBb_0$ and $Oc_0c_1$, the fact that $Oc_1x_1$ falls inside $Oc_1C$, and $c_1x_1=b_0B$, 
		we have $b_0B<c_1C$, which implies $b_2c_1 <BC$, because %$b_0B+ b_2c_1 <b_1C+c_1C$ 
		$Bb_2+ b_2c_1 <BC+c_1C$ and $Ab_2c_1 < ABC$. 
		
		Again using the fact that $O$ lies to the left of the bisector of $l_l$ and $\mathcal{B}$, with the same argument as above we have $Cc_3>b_0B$. 
		But then, since $b_0O>BO$, $Oc_0>OB$  (since $O$ lies to the left of the bisector) and $Oc_3>Oc_0$ (since $Oc_oc_3$ is a right triangle with $Oc_3$ as hypotenuse), we have $b_0c_3>BC$, which implies  $Ab_0c_3>ABC$, while $Ab_0c_3$ has also one internal angle bigger than $\frac {\pi}{2}$, contradiction with the optimallity of $ABC$ and Lemma~\ref{lem:angpstar}.
		
		Now suppose $BC$ is not vertical to $l_l$. But then the rotation of $l_r$ is allowed in both CW and CCW directions. 
		The proof is as above, using the fact that the side next to the largest angle in a triangle has the largest length. 
		Hence, we again reach $Ab_0c_3>ABC$, which gives a contradiction with the optimality of $ABC$. 
		
		The proof in the case where 	$O$ lies on the bisector of the angle $A$ is much simpler.  %The rotation of $l_r$ around $O$ makes the same triangles in scene of the perimeter. 
		Observe that among all possible rotations of $l_r$, the resulting maximum perimeter triangle is one of the two triangles at which further rotations of $l_r$ is not possible because either $l_r$ is blocked because it encounters another reflex vertex or $l_r$ reaches an endpoint of $l_l$. But in both of these cases, $l_r$ must be supported by two vertices, where one of them is $O$. 	The Lemma follows. 	
		\qed

		%But since $O$ lies to the left of the bisector which is passing through $AO$, from and the similarity of the triangles we have $b_0B<c_1C$, which implies $b_2c_1 <BC$, because $b_0B+ b_2c_1 <b_1C+Cc_1$ and $Ab_2c_1 < ABC$. Again using the fact that $O$ lies to the left of the bisector of $l_l$ and $\mathcal{B}$, we have $Cc_3>b_0B$, and since the distance function of $\ell$ is a convex function with a unique minimum, we have $b_0c_3>BC$, which implies  $Ab_0c_3>ABC$,~contradiction. 
		
		%For the case where the point $O$ lies at the bisector of $l_l$ and $\mathcal{B}$, the hyperbolic function is symmetric concerning the line passing through its minimum, which happens when $O$ lies at the midpoint of the line segment which has its endpoint at $l_l$ and $\mathcal{B}$. But since the function is a hyperbolic function, any other line segment which is passing through $O$ and has its endpoint at 	$l_l$ and $\mathcal{B}$ would have a longer length until further rotation would not pass through $O$.  %at which a maximum of $\mathcal{F}$  would be recognized.   
		
	\end{proof}

	%\vspace{-0.5cm}
	
		The relation between the triangles refers to their perimeter.

	\subsubsection*{Algorithm}
	In~\cite{cab,das2021largest} it is shown that a largest-area triangle inside $\mathcal{T}$ which has a base on the base of the terrain has the properties which we proved in Lemma~\ref{lem:alg}. The presented algorithms in~\cite{cab,das2021largest} enumerate all maximal triangles with these properties. Hence, we use the same algorithm, but we just report the largest-perimeter triangle instead of the largest-area triangle. For the case (i) of Lemma~\ref{lem:alg}, at which the top vertex lies at the boundary of the terrain, we apply Lemma 2 of~\cite{cab} to treat this case in $O(n \log n)$ time. 
	For the case (ii) we use the following results of~\cite{cab}.

	In Section~4 in~\cite{cab}, it is shown that each segment which is connecting two vertices of the terrain and may realize a side of $\Lambda^*_3$, would have two prolongations inside the terrain, at which in the prolongation in one direction the segment intersects the base, and in the other direction, the segment intersects the upper boundary of the terrain. Let $S$ denote all such segments. In Lemma 4 in~\cite{cab}, it is shown that the set of candidates of the left and the right sides of $\Lambda^*_3$ which consists of the prolongations of the segments in $S$, in both directions, can be computed in $O(n)$ time, as it uses the Guibas et al.~\cite{guibas1987linear} idea on how to compute the shortest path trees inside a simple polygon in $O(n)$ time. 
	Let $L$ denote the segments of $S$ with positive slopes, and let $R$ denote the segments of $S$ with negative slopes. 
	In the case where the prolongations of one element from $L$ and one element from $R$ lies inside the terrain, a triangle would be constructed which realizes a candidate of $\Lambda^*_3$. The set of all such intersection points (Lemma 5 in~\cite{cab}) realize the candidates of the top vertex of $\Lambda^*_3$, which can be computed in $O(n\log n)$ time (Lemma 4,5 in~\cite{cab}). Hence, we also construct all such triangles and report the one having the largest perimeter. 
	
	Consequently, a largest perimeter triangle contained in a 1.5D terrain of $n$ vertices which has a side on the base of the terrain can be computed in $O(n\log n)$ time. 
	In the following, we consider the case where a largest perimeter triangle has a vertex on the base of the terrain. 
	
\subsubsection*{$\Lambda^*_3$ having a vertex on $\mathcal{B}$ }	A largest perimeter triangle with a vertex on the base of the terrain has this property that the side opposite to the vertex on $\mathcal{B}$ has a maximal length since otherwise the perimeter can be improved. In Lemma~\ref{lem:longstick} we show that a maximal length line segment in $\mathcal{B}$ must be supported by two vertices of $\mathcal{T}$. Here we seek the maximal length line segments for which they are supported by two vertices of the terrain, but they do not necessarily have a prolongation which intersects the base of the terrain. Such maximal line segments are still supported by two vertices, and are a subset of the edges of the shortest path tree inside a simple polygon, that can be computed in $O(n)$ time, as discussed in detail in Section 3 in~\cite{cab}. 
	Thus we already know the two vertices for each potential largest perimeter triangle, and it remains to compute the optimal placement of the third vertex which lies on $\mathcal{B}$. According to Lemma~\ref{lem:ellipse}, the third vertex chooses an extreme placement on the line segment which it lies on, i.e., the leftmost or the rightmost possible locations on $\mathcal{B}$, at which the triangle does not intersect the terrain boundary. 
	Hence, for each maximal line segment, two maximal perimeter triangle needs to be considered, and we choose the largest one. Let $a,b$ denote the vertices of a maximal line segment $\ell$, and let $c \in \mathcal{B}$ be the third vertex of the maximal triangle having $\ell$ as a side. W.l.o.g., let $a$ be the left endpoint of $\ell$. 
Then  obviously either $c$ is a vertex of $\mathcal{B}$ or $ac$ is blocked by a reflex vertex $r$ (to stop further improvement of the perimeter), or it is passing through the terrain edge containing the vertex $a$. Handling the first and the last cases are straightforward. In the case where there is a  reflex vertex $r$ that blocks $ac$, we need to find $r$ precisely. But then, again, the line segment connecting $r$ to the leftmost vertex of the terrain edge containing $a$ is a subset of $L$, as it is supported by two terrain vertices and also has a vertex on $\mathcal{B}$~\cite{cab}. Hence, all such reflex vertices  can be computed in $O(n)$ time. The intersection of $\mathcal{B}$ and the prolongation of the line segment passing through $r$ and $a$  can be computed in constant time.  
%
	%In this case, we can consider the perimeter of  each potential triangle with two fixed vertices as a function of the location of the third vertex on the base of the terrain, and choose the placement which gives the maximum value. This can be done in constant time, but we need to ensure the sides adjacent to the selected vertex do not intersect the boundary of the terrain. The intersection of an edge of  a triangle and any edge other than the base of the terrain can be determined by ray shooting queries in $O(\log n)$ time.  If the  edge of  a triangle is also intersecting with an edge other than the base, we need shortened it by finding a new placement, at which the side is only touching a reflex vertex (for optimality of the length). Hence, we do a binary search for the edges of the terrain which lies between the current intersecting terrain edge and the base, and we remember the one which has 
	Consequently, a largest perimeter triangle which has a vertex on the base of the terrain can be computed in $O(n)$ time.

	\iffalse Our idea is to use Lemma~\ref{lem:alg}, to find the direction of an edge, and then construct the triangle entirely. 
	The edge which is supported by two vertices is either the left leg or the right leg, which we consider each case separately. 
	After fixing the direction of one side, we have the vertex of the triangle fixed at the base of the terrain. 
	For computing the other side, we compute the lower envelope of the vertices of the terrain on the opposite side, as only its edges are the candidates for determining the other leg of $P^*$. 
	Observe that any other candidate leg would intersect the terrain in an interior point. Since we have $O(n^2)$ candidates for $t$, 
	this gives an $O(n^2 \log n)$ algorithm, since the updates of the convex hull can be done in $O(\log n)$ time. This procedure is outlined in Algorithm~\ref{alg:exa}. \fi
	
	\begin{theorem}
		\label{thm:exacttri}
		The largest perimeter triangle contained in a 1.5D terrain of $n$ vertices can be computed in $O(n\log n)$ time.
	\end{theorem}
	
	%\subsection{Improvement}
	\iffalse
	\begin{algorithm}[h]
		\caption{Exact Triangle Algorithm}
		\label{alg:exa}
		\begin{algorithmic}[1]
			\Require{$T=\{P_1,\ldots,P_n\}$}
			\Ensure{An Triangle of maximum perimeter}
			\State{$P_{max}=0$}
			\State{$S_i$= the ordered set of vertices of $T$ around $p_i$ in CW}	
			\For{each $p_i\in T$}
			\For{each pair $p_j \in S_{i}$ $j \ne i$}
			\State{$l$= the supporting line of $(p_i,p_j)$}
			\State{$CH$= the lower envelop of all the vertices of $T$ to the left of $p_i,p_j$}
			\For{any edge $e$ of CH}
			\State{$P_{max}=$ the triangle constructed on the supporting line of $e$ and $l$}
			\State{$P^*=max(P^*,P_{max})$}
			\State{Update CH}
			\EndFor
			
			\EndFor
			\EndFor	
			\\ \Return{$P^*$}
		\end{algorithmic}
	\end{algorithm}
	\fi

	\begin{figure}[t]
		\centering
		\includegraphics[scale=0.5]{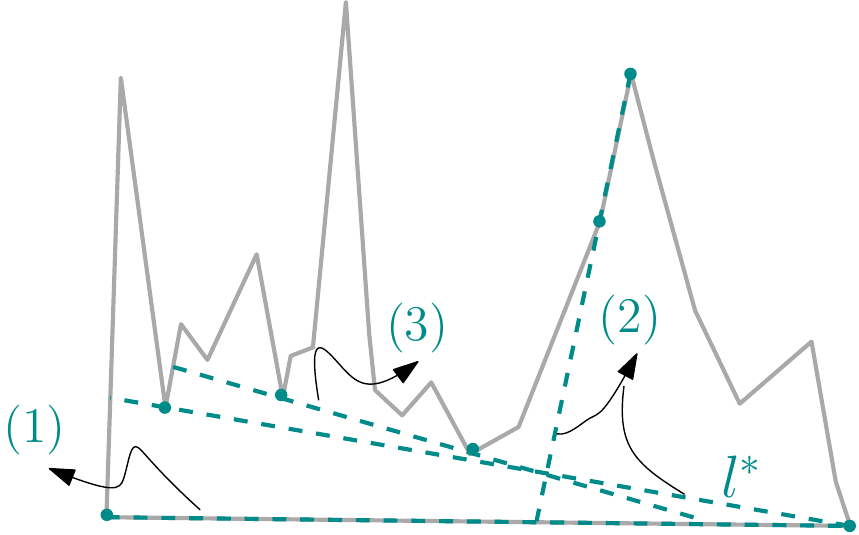}
		\caption{ Some examples for different candidates of the diameter of the terrain. The types of the events of  $l^*$ are denoted by arrows. }
		\label{fig:def}
	\end{figure}

	\subsection*{Case $k=2$: Diameter of the Terrain} 
	
	We define the {\em diameter} of a 1.5D terrain $\mathcal{T}$ as the longest line segment within $\mathcal{T}$ and denote it by $l^*$. 
	In the following, we discuss how we can compute $l^*$ in $O(n)$ time. 
	We design an approximation algorithm for Problem~\ref{prob:kgon} based on $l^*$.
	%We note that some of the proofs are moved to the Appendix due to the page limit. 
	%We note that all the missing proofs are in Appendix. 
	
	\begin{lemma} \label{lem:longstick}
		$l^*$ is either (1) supported by two convex vertices, 
		(2) supported by a reflex vertex and a convex vertex, or (3) supported by two reflex vertices of $\mathcal{T}$.   
	\end{lemma}

	\begin{proof}
		The longest line segment in the terrain must be blocked by the vertices of the terrain, otherwise, we still can increase its length by rotation or translation. 
		Case (2) have two different combinatorial structures, where  $l^*$ passes through an edge adjacent to a reflex vertex and ends up at a convex vertex, or $l^*$ originated from a convex vertex and passes through a reflex vertex. See Fig.~\ref{fig:def}.
		%The longest line segment inside the terrain which is passing through an edge adjacent to a reflex vertex, or the segment connecting two convex vertices of the terrain are maximal. 
		If the diameter is not a diagonal between two convex vertices, it must be blocked by at least one reflex vertex; if not we can improve the length by translation and/or rotation until the segment encounters a vertex to be blocked from further improvement on the lengths. 
		Let $r$ be the reflex vertex that blocks  $l^*$.  %If 	$l^*$ is not perpendicular to the intersecting edge of the terrain, if
		If we rotate  $l^*$ around $r$, in at least one of the CW or CCW directions the length of $l^*$ would be increased, until $l^*$ encounters another reflex vertex or a convex vertex that  further rotation is not possible. \qed %The lemma follows. \qed
	\end{proof}

	\iffalse 
	\begin{proof}
		The longest segment in the terrain must be blocked by some vertices otherwise we still can increase its length by rotation or translation. 
		Case (2) introduces two types of candidate solutions, where  $l^*$ passes through an edge adjacent to a reflex vertex, or $l^*$ originated from a convex vertex and passes through a reflex vertex. See Figure~\ref{fig:def}.
		Observe that the longest segment passing through an edge adjacent to a reflex vertex, or the segment connecting two convex vertices of the terrain is maximal. If the diameter is not a diagonal between two convex vertices, it must be blocked by at least one reflex vertex. 
		Let $r$ be the reflex vertex that blocks  $l^*$. If
		$l^*$ is not perpendicular to the intersecting edge of the terrain, if we rotate  $l^*$ around $r$, in at least one of the CW or CCW directions the length of $l^*$ would be increased, until $l^*$ encounters another reflex vertex to restrict further rotations, or it will encounter a convex vertex that further rotation is not possible.  The lemma follows. 
	\end{proof}

	\begin{figure}[t]
		\centering
		\includegraphics[scale=0.3]{bigcell}
		\caption{The union of four big cells (denoted by dashed segments) cover  the terrain entirely. The finer cells computed on the  lower bottom big cell (shaded) are denoted by solid lines.  }
		\label{fig:cells}
	\end{figure}
	
	\fi

	%\vspace{-0.8cm}

	%\textbf{Diameter Algorithm} 
	%See Figure~\ref{fig:def}(b).

	For computing $l^*$, the existence of an $O(n^2 \log n)$ time algorithm by considering any pair of vertices and ray shooting queries is obvious. In the following, we discuss the improvement of this running time. 
	
	\begin{theorem} \label{thm:diameteroft}
		There is a solution to the diameter of the terrain which has a vertex on  $\mathcal{B}$ and can be computed in $O(n)$ time. 
	\end{theorem}

	\begin{proof}
		From Lemma~\ref{lem:longstick} we know any maximal length line segment $l^*$ in $\mathcal{T}$ is supported by two vertices of $\mathcal{T}$. If $l^*$ does not have any vertex on $\mathcal{B}$, we can translate it downward until one of its vertices lies at $\mathcal{B}$, and then we either can improve the length or keep it unchanged if its blocked from further changes. Hence, the set $S$ discussed in Section~\ref{sec:exactalg} is a superset for all maximal length line segments in the terrain.  In Lemma 4 in~\cite{cab}, it is shown that the set $S$ which consists of the prolongations of all the maximal length line segments passing through two terrain vertices can be computed in $O(n)$ time. \qed
		%The result follows. 
	\end{proof}

	\section{An FPTAS for $\Lambda^*_k$}

	%to $O(n^2)$ by considering the arrangement of the supporting lines of the edges of the terrain. 
	%We discuss a more efficient algorithm. The optimal solution to case (2) can be computed in $O(n \log n)$ time. For computing the optimal solution to case (3) of Lemma~\ref{lem:longstick}, a each reflex vertex $r_i$, we find two other reflex vertices $r^{cw}_i$ and $r^{ccw}_i$ that gives the longest length at both CW and CCW directions. 
	We first discuss $k=3$. Any  triangle which has the diameter of the terrain as a side gives a 3-approximation to $\Lambda^*_3$. We use this simple observation to design our  FPTAS for Problems~\ref{prob:triangle},\ref{prob:kgon}. 
	We first explain the algorithm for $k=3$ and discuss the extension afterwards. 
	%This gives a simple $O(n^2 \log n)$ time approximation algorithm which we use to design our FPTAS. 
	\iffalse
	\begin{algorithm}[t]
		\caption{PTAS for the triangle}
		\label{alg:ptas}
		\begin{algorithmic}[1]
			\Require{$T=\{P_1,\ldots,P_n\}, \epsilon>0$}
			\Ensure{A $(1-\epsilon)$-approximation of the triangle of maximum perimeter}
			\State{$P_{max},P^*=0$}	
			\State{$P'$=approximated triangle}
			\For{each big cell $c$}
			\State{Compute the finer cells of side length $P'\epsilon$}
			\For{each triple of finer cells}
			\State{Compute the pairwise visibility regions on the cor. sides of a pair of cells}
			\If{$\Phi \ne \null$}
			\State{$P_{max}=$ the triangle constructed on the endpoints}
			\EndIf
			\State{$P^*=max(P^*,P_{max})$}

			\EndFor
			\EndFor	
			\\ \Return{$P^*$}
		\end{algorithmic}
	\end{algorithm}
	
	\fi

	Let let $\Lambda$ be a $3$-approximation of $\Lambda_3^*$. Also, let $|\Lambda|$ denote the perimeter of $\Lambda$.  Consider a grid of big cells of side length $6 |\Lambda|$. 	
	Let the bottom left corner of a big cell in the grid lies at the leftmost vertex of $\mathcal{T}$, and let $(0,0)$ denote its coordinates. Consider three copies of this big cell with the bottom left corners at coordinates $(3|\Lambda|,0)$, $(0,3|\Lambda|)$ and  $(3|\Lambda|,3|\Lambda|)$, respectively.  Since $\mathcal{T}$ is monotone, the diameter of the terrain at least equals the length of $\mathcal{B}$, and the longest vertical line segment inside the terrain is smaller than the diameter of the terrain,  the union of the four big cells covers the terrain entirely.

	\begin{lemma} \label{lem:bigcell}
		$\Lambda^*_3$ is contained in one of the big cells entirely. 
	\end{lemma}	
	\begin{proof}
		The side length of a big cell, say $C$ is $6 |\Lambda|$. Since any triangle has at most one obtuse angle, in the worst case, two consecutive edges of $\Lambda$ fulfils at most a 1/9 fraction of the width of $C$. Hence, any feasible solution to $\Lambda$ entirely lies within $C$. On the other hand, $\Lambda$ gives a 1/3 approximation for $\Lambda^*_3$, and $\Lambda^*_3$ does not have any obtuse angle, but any of its sides is at most trice of a side of $\Lambda$. Hence, $\Lambda^*_3$ also entirely lies within a big cell.  \qed
	\end{proof}
	%%%\begin{lemma} A 1.5D terrain $\mathcal{T}$ with the largest perimeter inscribing triangle of size at most $3P$ can be covered by 4 squares of side length $6 P$. \end{lemma}

	From the construction of the grid, each edge of $\mathcal{T}$ intersects at most all the four big cells. 
	We  decompose each big cell by $O(\frac 1 \epsilon)$ finer cells of side length $\epsilon |P|$, for a given $\epsilon>0$.  Let $X$ denote the set of finer cells. Observe that the total complexity of the intersection of $X$ and $\mathcal{T}$ is in $O(n)$ since each edge of $\mathcal{T}$ is intersecting with at most four big cells. So computing the set $X' \subseteq X$ of finer cells that are intersecting with $\mathcal{T}$ takes $O(n)$ time. 
	We consider any of the four big cells independently. 
	
	Let $c_i,c_j$ and $c_k$ be any triple of finer cells in $X'$, all in a specific big cell $C$. 
	For two line segments $s^i_p \subset c_i$ and $s^j_q \subset c_j$, we consider the intersection of the visibility polygons by assuming an edge guard at $s^i_p$ and another one at $s^j_q$. 
	However, we do not need to compute the visibility polygons explicitly. 
	Let $V^{ji}_q$ 
	denote the ranges on the corresponding side of $c_i$ that are visible to $s^j_q$.  
	We are interested in determining whether there are three line segments in $\mathcal{T} \cap C$ whose endpoints lie at the intersection of the pairwise visibility ranges of the segments $s^i_p$, $s^j_q$ and $s^k_r \subset c_k$. This happens when  the intersection of the pairwise visibility ranges of $s^i_p$, $s^j_q$ and $s^k_r$ have a non-empty intersection, i.e., $V^{ij}_p \cap V^{kj}_q\ne \emptyset$, $V^{ji}_q \cap V^{ki}_r\ne \emptyset$  and $V^{ik}_p \cap V^{jk}_r\ne \emptyset$,  as illustrated in Fig.~\ref{fig:cells}.

	\begin{figure}[t]
		\centering
		\includegraphics[scale=0.8]{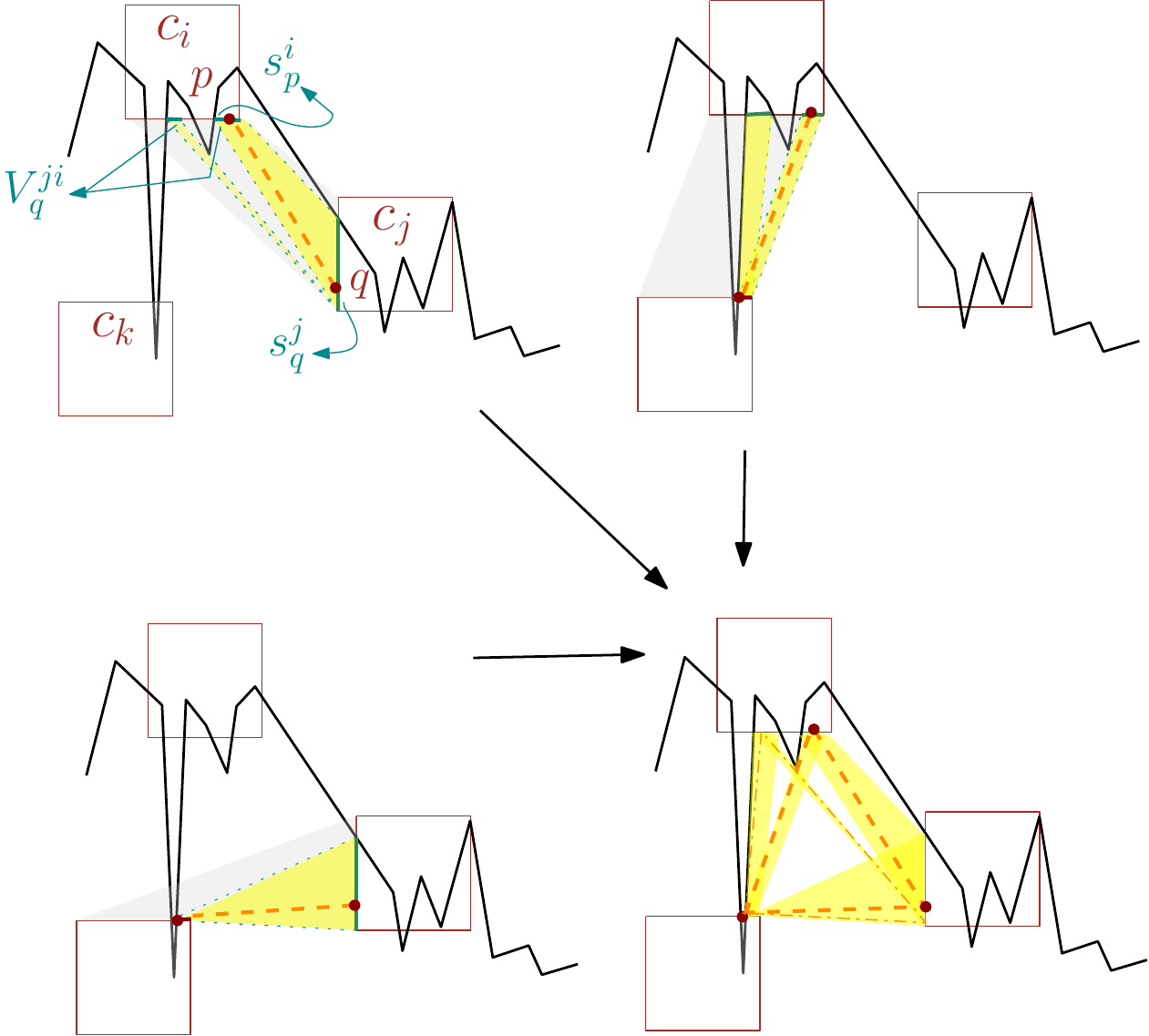}
		\caption{Computing the visibility ranges on a triple of finer cells that are introducing a candidate triangle for $\Lambda^*_3$ (its sides denoted by dash). The visibility ranges are shown in yellow. 
			This triple of segments has a non-empty intersection for the pairwise visibility ranges of the segments. Here, there also exists another combinatorially different triangle (in dash-dotted) on these three finer cells.   }
		\label{fig:cells}
	\end{figure}
	%The intersection of the parallelogram constructed on the sides of each pair of cells and $T$ makes a set of polygons, which we check separately. 
	%\vspace{-0.3cm}
	
	\subsection{Preparation for Using the Data Structures} \label{app:ds}
	
	The visible parts of a given segment $s$ from any other segment, as well as the visible parts of every
	segment from $s$, can be computed in linear time and space using the shortest
	path algorithm~\cite{guibas1987linear}. This imposes an extra $O(n^2)$ time in the running time of the algorithm. Next, we try to improve this running time using the technique presented  in~\cite{hall2006finding} for computing the visibility between intervals on the edges of a simple polygon.   Here we recall the main technique of the presented algorithm in~\cite{hall2006finding} for the self-completeness of our algorithm.
	
	Let $c_i,c_j$ and $c_k$ be three finer cells in a big cell. 
	For any three distinct sides of the cells  $c_i,c_j$ and $c_k$, we first make the corresponding visibility ranges, i.e., the maximal length line segments on a side of a finer cell that lies inside the terrain entirely. 
	We refer to each visibility range as an {\em interval}. Let $s^i_1,\ldots,s^i_{n_i}$, $s^j_1,\ldots,s^j_{n_j}$ and $s^k_1,\ldots,s^k_{n_k}$ denote the sequence of the intervals on $c_i$, $c_j$ and $c_k$, respectively, such  that  they lie within $\mathcal{T}$. 
	
	Let $V^{ij}_p$ be the constructed {\em visibility range} for the interval $s^{i}_p$, by considering an edge guard on it, with respect to the intervals on $c_j$.  %$\in\{ s^i_1,\ldots,s^i_{n_i},s^j_1,\ldots,s^j_{n_j},s^k_1,\ldots,s^k_{n_k}\}$.  
	If one or both endpoints of $V^{ij}_p$ lie in any visibility ranges of $c_j$,    $s^{i}_p$ sees any of these intervals. This can be determined in logarithmic time. %For the cell $c_j$, at most 4 intervals  (at most two intervals for the visibility ranges on one cell, that makes at most four intervals for two cells) need to be considered. 
	
	%To make the computation easier, we shrink any of the sub-ranges of $V^{ij}_p$  such that only those intervals that are fully contained in it remain. This way
	We can represent $V^{ij}_p$ by the pair of indices corresponding to the %leftmost interval and the rightmost 
	intervals that are contained in it. 
	Then we can do range searching queries and 1D windowing queries to determine the intervals which have a non-empty intersection between their visibility ranges. %Our idea of using the data structures for computing the visibility between intervals on the edges of the terrain was first introduced  and also discussed in detail in~\cite{hall2006finding},  

	For each pair of $s^i_p$ and $s^j_q$, for efficiently determining whether there is any segment that lies within $\mathcal{T}$ entirely, and with endpoints at these intervals, we use a combination of two data structures. %(see Appendix~\ref{app:ds} for the details). 
	
	For two points $p \in s^i_p$ and $q \in s^j_q$, the line segment $\overline{pq}$ must lie in the visibility range of both $s^i_p$ and $s^j_q$. The visibility range of any of $s^i_p$ and $s^j_q$ can be computed in $O(\log n)$ time by performing shortest path queries on preprocessed connected components achieved in $O(n)$ time~\cite{guibas1989optimal}, and checking whether they can see each other or not can be done at the same time by simple visibility queries inside simple polygons. 
	We first make a 1D range tree at the indices of the intervals in $c_i$ in $O(n_i \log n_i)$ time, and of the height $O(\log n_i)$. Then we make an interval tree at each node of the constructed range tree. Thus our data structure has a size $O(n_i \log n_i)$ and takes $O(n_i \log^2 n_i)$ time (spending $O(n_i \log n_i)$ time at each level of the range tree). 
	The queries are is there any interval $s^i_p$ in $c_i$ which is visible from an interval $s^j_q$ in $c_j$. This can be answered in $O(\log n_i)$ time. Performing $n_j$ queries like this gives us the running time $O((n_i+n_j) \log^2 n_i)$ for each pair of cells in a big cell. Observe that  $n_i,n_j \in O(n)$. 
	
	\begin{lemma}
		For each pair of finer cells, computing all pairwise visible intervals takes $O(n \log^2 n)$ time~\cite{hall2006finding}. 
		\end{lemma}
	
	There are $O(\epsilon^{-2})$ cells, and we consider any triple of cells. %in this set. 
	For any triple of intervals that the intersection of the pairwise visibility regions is non-empty, we compute the largest perimeter triangle as below.
	%This would give an approximation for $P^*$, but, as discussed before, the running time is not efficient. 
	%as bellow. 

	%\vspace{-0.5cm}
	
	\subsection{A Largest perimeter triangle on  three intervals}
	
	For any three points on the boundary of terrain that define a triangle that lies inside the terrain entirely, it is necessary
	and sufficient that they are pairwise visible. 
	Now we have three intervals such that they are pairwise visible, and we seek the largest perimeter triangle having its vertices on these intervals. 
	
	Any triple $s^i_p,s^j_q,s^k_r$ of segments that may contain the vertices of $\Lambda^*_3$ have a horizontal or vertical direction. 
	Suppose the locations of two points  $p,q$ on their segments, say $s^i_p,s^j_q$  are fixed, and the position of $r \in s^k_r$ is a variable. 
	The function that
	describes the longest perimeter changes as a symmetric hyperbolic function with a unique minimum. Thus there always exists one direction in which a point can be moved such that the size has
	not decreased. In the following, we prove it formally.

	\begin{figure}[t]
		\centering
		\includegraphics[scale=0.7]{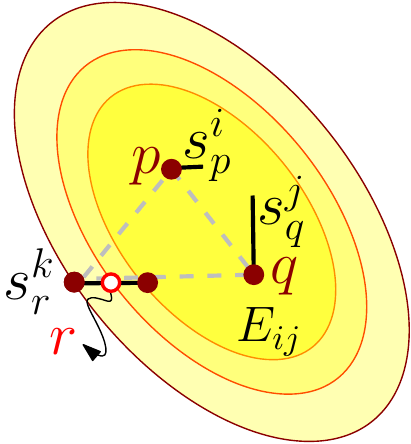}
		\caption{ $E_{ij}$ is the constructed ellipse at $\Lambda^*_3=\bigtriangleup pqt$. If there is at least one line segment $s^k_r$ that does not contribute an endpoint to $\Lambda^*_3$, at least one of the two ellipses which are constructed at $p,q$ and one of the endpoints of $s^k_r$ contains $E_{ij}$ entirely. (Note that two of the ellipses coincide if $s^k_r$ is parallel to $\overline{pq}$).  }
		\label{fig:ellips}
	\end{figure}

	\begin{lemma} \label{lem:ellipse}
		For any triple $s^i_p,s^j_q$ and $s^k_r$ of segments, the largest perimeter triangle $\Lambda^*_3$ with one vertex at each segment would have its vertices at the endpoints of the segments. 
	\end{lemma}	
	
	\begin{proof}
		Suppose, by contradiction, that $\Lambda^*_3=\bigtriangleup pqr$ has at least one vertex that is not at an endpoint. 
		W.l.o.g, let  $p\in s^i_p$ and $ q \in s^j_q$ be the vertices of $\Lambda^*_3$ that lie at the endpoints. Consider the ellipse $E_{ij}$ which is passing through $r$ and with  $p$ and $q$ as its foci, see Fig.~\ref{fig:ellips}.  
		Consider the two (possibly one) ellipses that are passing through the endpoints of the line segment $s^k_r$ with the same foci as $E_{ij}$.
		Then since $r$ is an interior point of $s^k_r$, $E_{ij}$  entirely lies within one of the constructed ellipses. Because of this,  
		moving 
		$r$ along $s^k_r$ to at least one of the endpoints of $s^k_r$ would strictly increase the perimeter of $\Lambda^*_3$; contradiction. \qed
	\end{proof}
	
	\iffalse \begin{proof}
		Suppose, by contradiction, that $P^*=\bigtriangleup pqk$ has at least one vertex that is not at an endpoint. 
		W.l.o.g, let  $p\in s^i_p$ and $ q \in s^j_q$ be the vertices of $P^*$ that lie at the endpoints. Consider the ellipse $E_{ij}$ which is passing through $k$ and with  $p$ and $q$ as its foci, see Fig.~\ref{fig:ellips}.  
		Consider the two (possibly one) ellipses that are passing through the endpoints of the line segment $s^k_r$ with the same foci as $E_{ij}$.
		Then since $k$ is an interior point of $s^k_r$, $E_{ij}$  entirely lies within one of the constructed ellipse. Then obviously 
		moving 
		$k$ along $s^k_r$ to at least one of the endpoints of $s^k_r$ would increase the perimeter of $P^*$, contradiction. 
	\end{proof} \fi
	
	%%%See Figure~\ref{fig:ellips} for the sketch of the proof. Hence for any triple of fixed edges, the largest perimeter triangle with one vertex at any of those edges can be computed in $O(1)$ time, and 
	Thus, a $(1-\epsilon)$-approximation of the largest perimeter triangle in $\mathcal{T}$ can be computed in $O(\epsilon^{-6}n\log^2n)$ time, however, this is not efficient as the problem has an $O(n \log n)$ time exact algorithm. But this algorithm can be generalized for computing largest area/perimeter polygons of $k>3$ vertices. 
	\section{Extension to Convex Polygons of at most $k$ Vertices} \label{sec:pkgon}
	
	The extension of the presented algorithm to convex polygons of $k>3$ vertices is straightforward. We consider any set of $k$ intervals in $k$ finer cells instead of three intervals, where all the $k$ intervals lie inside the terrain and are pairwise visible. 
	
		The non-empty intersection condition is necessary again to find a largest perimeter {\em convex} polygon of $k$ vertices that entirely lies inside $\mathcal{T}$, however, in the algorithm, we may report a convex polygon of less than $k$  vertices if it has a larger area/perimeter than any convex $k$-gon with vertices at the selected $k$~intervals. 
	It remained to find the optimal placements of the points on each set of $k$ segments that admits a non-empty set for the intersection of the pairwise visibility ranges.
	
	%Here, the finer cells would not necessarily lie inside a big cell. 
	%Similar to the case where $k=3$, we need to compute the pairwise visibility conditions of the visibility ranges for all the selections of the pairs of $k$ intervals. 
	
		\subsection{Perimeter Measure}  \label{sec:akgon}

	L\"offler and van Kreveld have shown that if there is a set of $O(n)$ squares or line segments of two directions, the maximum perimeter convex polygon that is intersecting all of them and has at most one vertex at any of them can be computed in $O(n^{10})$ time~\cite{loffler2010largest}, and the optimal solution always chooses the endpoints of the segments. 
	In our problem, the intervals have horizontal or vertical directions since they lie at the sides of the finer cells. 
	However, in the algorithm of~\cite{loffler2010largest}, it is necessary for the segments to be disjoint. On the other hand, we do not have the restriction of choosing exactly one point from each interval. But we can adjust our intervals to satisfy the input conditions of the algorithm presented in~\cite{loffler2010largest}. 
	
	We need to allow an interval to contribute to both of its endpoints. 
	Hence, for each interval, we send two tiny length sub-intervals to the algorithm, such that each of which is contained in the original interval and contains exactly one of the endpoints of the original interval. This way, a horizontal line segment in our problem is transformed into two tiny length horizontal line segments. With a careful selection of the length of the tiny intervals which should be a fraction of $\epsilon$ (to keep the approximation factor unchanged), the tiny intervals would not intersect each other and the selection of each of their endpoints does not change the combinatorial structure of the resulting polygon. Hence, from now on, we work with the tiny intervals instead of the original intervals. 
	
	We note that we shortened the intervals because we know from~\cite{loffler2010largest} that only the endpoints of the tiny length intervals contribute to the optimal solution.  %, however, for the sufficiently small length of the tiny segments, the selection of either endpoint would not change the combinatorial structure of the solution. 
	Also notice that we cannot simply consider the endpoints (degenerate segments) since we miss the pairwise visibility information by only considering the points. 
	%Also, we still need to keep that which segment has a vertical or a horizontal direction for making a combinatorially valid solution that lies inside the terrain entirely. 

	Any convex at most $k$-gon $\Lambda$  which has the diameter of the terrain as an edge gives a $1/k$ approximation for the largest perimeter convex $k$-gon in the terrain. Hence, we would have four big cells of side length $2k|\Lambda|$ to cover the terrain~\footnote{The extension of the Lemma~\ref{lem:bigcell} to this case follows from the convexity of the $k$-gon and is straightforward.}. 
	So we use the presented dynamic program algorithm in~\cite{loffler2010largest}.
	We just need to ensure the next vertex from the current tiny interval is visible from the current vertex. A natural correction is we keep a refined list of tiny intervals for each of the tiny intervals, where only visible tiny intervals are included in the refined list of each tiny interval. This increases the space complexity by $O(n)$, but the running time remains unchanged.

	Hence, we can use the presented algorithm in~\cite{loffler2010largest} for computing a largest perimeter polygon of at most $k$ vertices with vertices on $k$ intervals (see \cite{loffler2010largest} for details). Notice that the constructed polygon has at most $2k$ vertices, but the adjustment for having at most $k$ vertices is straightforward. 
	Consequently, we achieve an $O(k^{10}\epsilon^{-2k}n \log ^2 n)$ time algorithm that approximates the largest perimeter  convex polygon of at most $k$ vertices within a factor $(1-\epsilon)$. 
	
	%\begin{lemma} Above algorithm approximates the maximum perimeter triangle within a factor $1-\epsilon$. 	\end{lemma}
	
	\begin{theorem} \label{thm:kgon}
		One can compute a $(1-\epsilon)$-approximation of the largest perimeter $k$-gon in a 1.5D terrain $\mathcal{T}$ of $n$ vertices in $O(k^{10}\epsilon^{-2k}n \log ^2 n)$ time and $O(n^2)$ space. %A $\frac 1 3$ approximation can be computed in $O(n)$ time.  
	\end{theorem}
	
	\begin{proof}
		Computing the largest perimeter polygon on $k$ intervals takes $O(k^{10})$ time. There are $O(\epsilon^{-2k})$ sets of $k$ intervals in $O(\epsilon^{-2})$ finer cells. Considering a pairwise visibility query takes $O(n \log^2 n)$ time. Hence, we have $O(k^{10}\epsilon^{-2k}n \log ^2 n)$ computations and we remember the largest computed polygon which has at most $k$ vertices. The space complexity of the algorithm of~\cite{loffler2010largest} is $O(n)$. Hence, in our problem, the space complexity is $O(n^2)$.  \qed
	\end{proof}
	
	%\vspace{-0.6cm}
	
	\subsection{Area Measure}  \label{sec:akgon}
	In this section, we show that our algorithm is extendable to the area measure. We first prove that if two triangles have a $(1-\epsilon)$-approximation for the ratio of their corresponding sides, the ratio of their area is also within a factor $(1-\epsilon)$.  
	
	\begin{lemma}	Suppose the ratio between the length of corresponding edges of two triangles $\bigtriangleup abc$ and $\bigtriangleup a^*b^*c^*$  is  $(1-\epsilon)$. Then the area of $abc$ is within a factor $(1-\epsilon)$ of the area of $a^*b^*c^*$.
		
	\end{lemma}		
	
	\begin{proof}
		From the statement of the lemma, $\frac {ab}{a^*b^*}=\frac {bc}{b^*c^*}=\frac {ac}{a^*c^*}$. This implies that $\bigtriangleup abc$ and $\bigtriangleup a^*b^*c^*$ are similar. A basic theorem in the similarity of the triangles states that for two similar triangles with similarity ratio $(1-\epsilon)$, the corresponding heights is also within a factor $(1-\epsilon)$, and the ratio of the areas equals $(1-\epsilon)^2$, which is still a $(1-\epsilon)$ approximation. \qed \end{proof}
	
	\paragraph{Algorithm.}
	Similar to the largest-perimeter at most $k$-gon, we should find the optimal placements of the points on each set of $k$ intervals that admits a non-empty set for the intersection of all the pairwise visibility ranges. 
	In~\cite{loffler2010largest}, an $O(n^7)$ time algorithm is designed to compute the largest-area convex hull with vertices given from a set of disjoint axis-aligned squares of arbitrary size, where each square contributes at most one vertex. 
	The authors
	proved that only two corners of each square can contribute a vertex to the optimal solution, and they reduced their problem to the problem of finding the largest-area convex polygon on a set of line segments of four different orientations; horizontal, vertical, and diagonal, where the diagonal orientations come from the diagonal of the squares. 
	%
	%The disjointness assumption of the squares is to reduce six possible orientations to five.  
	In our problem, we only have two orientations for the intervals, so we can use the same algorithm to compute the largest area convex hull of at most $2k$ vertices, where the vertices lie at horizontal or vertical tiny intervals, as discussed in the previous section. This gives us an $O(k^{7} \epsilon^{-2k} n \log^2 n)$ time algorithm for having a $(1-\epsilon)$-approximation algorithm. 
	\begin{theorem}
		\label{thm:areatri}
		One can compute a $(1-\epsilon)$-approximation of the largest area $k$-gon in a terrain $\mathcal{T}$ of $n$ vertices in $O(k^{7}\epsilon^{-2k}n\log^2n)$ time and $O(n^2)$ space.
	\end{theorem}	
	
	\section{Discussion}
	We discussed several algorithms for computing a largest area or perimeter polygon of bounded complexity $k$. For $k=2,3$, we discussed exact algorithms for the perimeter measure. Furthermore, for $k=3$ our running time is the same as the running time of the algorithm presented in~\cite{cab} for computing the largest area triangle. The efficiency of the algorithms for $k=3$ in both the area and the perimeter measures remained open. Also, the problem of finding a largest area/perimeter convex polygon of at most $k$ vertices for general $k$ remained open.  

	\bibliographystyle{abbrv}
	%\vspace{-0.35cm}
	\bibliography{refs}
	
\end{document}